\documentclass[onecolumn,12pt]{IEEEtran}
\IEEEoverridecommandlockouts                              
\usepackage{psfrag}
\usepackage{setspace}
\usepackage{amsmath}%
\usepackage{amsfonts}%
\usepackage{amssymb}%
\usepackage{graphicx}
\graphicspath{{figures/}}
\usepackage{subfigure}
\usepackage{cite,url}
\usepackage{epstopdf}
\usepackage{comment}
\usepackage{SDC_macros}

\usepackage{color}

\graphicspath{{figures/}}
\newcommand{\co}{{\rm co}}
\newcommand{\bcmt}{}
\newcommand{\sdcRemoved}[1]{}

\renewcommand{\eal}{\end{align}}
\title{Indirect-adaptive Model Predictive Control\\ for Linear Systems with Polytopic 
Uncertainty}
\author{Stefano Di Cairano
 \thanks{S. Di Cairano is with Mitsubishi Electric Research Laboratories,
Cambridge, MA, email: {\tt
dicairano@ieee.org}
}}

\begin{document}
\maketitle


\begin{abstract}
We develop an indirect-adaptive model predictive control algorithm for uncertain linear 
systems subject to constraints. The system is modeled as a polytopic linear parameter varying 
system where the  convex combination vector is constant but unknown. Robust constraint 
satisfaction is obtained by constraints enforcing a  robust control invariant. The terminal cost 
and set are constructed from a 
parameter-dependent Lyapunov function and the associated control law. The proposed design 
ensures robust constraint satisfaction and  recursive feasibility, is 
input-to-state stable with respect to the parameter estimation error and it only requires the online 
solution of quadratic programs.  
 \end{abstract}

\section{Introduction}
%
%
%

In Model Predictive Control (MPC)~\cite{RM09} the 
prediction model is exploited for evaluating feasibility and performance of the sequence of 
actions to be selected by the controller. In several cases, some of the model parameters 
are uncertain   at   design time, especially when a  
controller is deployed to multiple instances of the plant, such as in automotive, factory 
automation, and aerospace applications~\cite{dicairano12,MPCautoSurvey12}, where 
the control algorithm and its auxiliary functions need also to have low complexity and 
computational effort, due to stringent cost, timing, and validation requirements.

For the cases when model parameters are uncertain, robust MPC methods have been 
proposed, see, e.g.,~\cite{KBM96,cuzzola02,mao03,tubeMPC,fukushima2005}. Some of the 
limitations of these methods are either in the computational cost, due to solving linear matrix 
inequalities (LMIs) at each control step~\cite{KBM96,cuzzola02,mao03}, or in  
applying online to additive disturbances~\cite{tubeMPC}, or in imposing limitations on cost 
function and terminal set~\cite{fukushima2005}. These limitations often arise due to 
considering the challenging case where uncertain parameters are constantly changing during 
system 
operation. 

However, when the parameters are unknown but are constant or slowly varying over time, an 
alternative approach is to learn their values, resulting in adaptive control techniques that ensure 
safe operation during the learning phase, and improve performance, for instance in terms of 
stability or tracking, as the learning proceeds. Adaptive MPC algorithms 
have been recently proposed based on different methods, such as a 
comparison 
model~\cite{kim2008}, min-max approaches with open-loop 
relaxations~\cite{adetola2009adaptive}, learning of constant 
offsets~\cite{aswani2013provably}, and set membership 
identification~\cite{tanaskovic2014}.  
Another class of adaptive MPC 
algorithms  focuses on  ``dual objective control'',  i.e., controlling the system while 
guaranteeing 
sufficient excitation  for identifiability, see, 
e.g.,~\cite{marafioti2010persistently, rathousky2013mpc,WD14}, and references therein.

In~\cite{dic15}, a MPC design  allowing for the prediction 
model to be adjusted after deployment was proposed. In this  paper 
we propose a MPC design that  operates concurrently with a parameter estimation scheme, thus 
resulting in an indirect-adaptive MPC (IAMPC) approach, that retains constraints satisfaction 
guarantees and certain stability properties. 
Motivated by the case of the unknown but constant (or slowly varying) parameters and by the need 
to keep computational burden small for application in fast systems equipped with low cost 
microprocessors~\cite{dicairano12,MPCautoSurvey12}, here  we do not 
seek robust stability as in robust MPC, but rather robust  constraint satisfaction and an 
input-to-state stable (ISS) closed-loop with respect to 
the estimation error. ISS will hold with only minimal assumptions on the estimates, and if the 
correct parameter value will be eventually estimated (possibly in finite time, see, 
e.g.,~\cite{adetola2008finite}), by the definition of ISS the 
closed-loop will become asymptotically stable (AS). Constraint satisfaction holds even 
if the parameters keep changing.  

For uncertain systems represented as polytopic linear difference inclusions (pLDI) we design a 
parameter-dependent quadratic 
terminal 
cost and a robust 
terminal constraint using a parameter-dependent Lyapunov function  
(pLF)~\cite{daafouz01}  and 
its corresponding stabilizing control law. Robust constraint satisfaction in presence 
of parameter estimation error is obtained 
by enforcing robust control invariant set constraints~\cite{BM08}. A  parameter 
prediction
update law is also designed to ensure the desired properties. The proposed 
IAMPC allows uncertainty in the system dynamics, as opposed to additive disturbances/offsets 
as in~\cite{tubeMPC,aswani2013provably}, and only solves quadratic 
programs (QPs), as opposed to robust MPC methods that require the 
online 
solution of LMIs~\cite{cuzzola02,mao03}. 

The paper is structured as follows. After the preliminaries in 
Section~\ref{sec:prelim}, in 
Section~\ref{sec:unconstr} we 
design 
the   cost function  
that results in the
unconstrained IAMPC to be ISS with respect to the estimation error. For 
constrained IAMPC,  in Section~\ref{sec:constr} we first design the terminal set that 
guarantees that the nominal 
closed-loop is AS. Then, we design robust 
constraints that ensure that the system constraints are satisfied and the IAMPC remains 
feasible even in presence of parameter estimation error. In Section~\ref{sec:IAMPC} we combine 
the cost function and the 
constraints with a parameter estimate prediction update and we briefly discuss the required 
properties of the estimator, hence describing the complete IAMPC. In 
Section~\ref{sec:sims} we show a numerical example. 
Conclusions and future developments are discussed in Section~\ref{sec:conclus}.

{\em Notation:} $\rr$, $\rr_{0+}$, $\rr_+$, $\zz$, $\zz_{0+}$, $\zz_+$ are the sets of real, 
nonnegative real, positive real, and integer, nonnegative integer,
positive
integer numbers. We denote interval of numbers using notations like 
$\zz_{[a,b)}=\{z\in\zz:a\leq 
z < 
b\}$. $\co\{\XX\}$ denotes the convex hull of the set $\XX$.
For vectors, inequalities are intended componentwise, while for 
matrices indicate (semi)definiteness, and $\lambda_\rmin(Q)$ denotes the smallest 
eigenvalue of $Q$.  By $[x]_i$
we denote the $i$-th component of vector $x$, and by $I$  and $0$ the identity and the 
``all-zero''
matrices of appropriate dimension. $\|\cdot\|_p$ denotes the $p$-norm, and 
$\|\cdot\|=\|\cdot\|_2$. For a discrete-time signal $x\in\rr^n$ with sampling 
period 
$T_s$, 
$x(t)$
is the state a sampling instant $t$, i.e., at time $T_st$, $x_{k|t}$ denotes the
predicted value of  $x$ at sample $t+k$, i.e., $x({t+k})$, based on data at sample $t$, and 
$x_{0|t}=x(t)$. A function $\alpha:\rr_{0+}\rar\rr_{0+}$ is of class $\KK $ if it is 
continuous, strictly increasing,  $\alpha(0)=0$; if in addition 
$\lim_{c\rar\infty}\alpha(c)=\infty$, $\alpha$ is of class $\KK_\infty$.

\section{Preliminaries and Problem Definition}\label{sec:prelim}

Details on the following standard definitions and results  are, e.g., in~\cite[Appendix B]{RM09}.

\begin{definition}
Given  $x(t+1)=f(x(t),w(t))$, $x\in\rr^n$,  $w\in\WW\subseteq \rr^d$, a set $\SS\subset\rr^n$ is 
robust positive invariant (RPI) for $f$ iff 
for all $x\in\SS$, $f(x,w)\in\SS$, for all $w\in\WW$. If $w=\{0\}$, $\SS$ is called positive 
invariant (PI).\cvd
\end{definition}

\begin{definition}\label{def:RCI}
Given  $x(t+1)=f(x(t),u(t),w(t))$, $x\in\rr^n$, $u\in\UU\subseteq \rr^m$, $w\in\WW\subseteq 
\rr^d$, a set $\SS\subset\rr^n$ is 
robust control invariant (RCI) for $f$ iff 
for all $x\in\SS$, there exists $u\in\UU$ such that $f(x,u,w)\in\SS$, for all $w\in\WW$. If 
$w=\{0\}$, $\SS$ is called control
invariant (CI).\cvd
\end{definition}

\begin{definition} Given $x(t+1)=f(x(t))$, $x\in\rr^n$, and a PI set $\SS$ for $f$, $0\in\SS$, 
a function $\VV:\rr^n\rar\rr_{0+}$ 
such that there exists $\alpha_1,\alpha_2,\alpha_\Delta\in\KK_\infty$ such that
$\alpha_1(\|x\|)\leq \VV(x) \leq \alpha_2(\|x\|)$, $\VV(f(x))-\VV(x)\leq -\alpha_\Delta(\|x\|)$ for 
all 
$x\in\SS$  is a Lyapunov function  for  
$f$ in $\SS$.\cvd
\end{definition}

\begin{definition}
Given $x(t+1)=f(x(t),w(t))$, $x\in\rr^n$, $w\in\WW\subseteq \rr^d$, and a RPI set $\SS$ for 
$f$, 
 $0\in\SS$, 
a function $\VV:\rr^n\rar\rr_+$ 
such that there exists $\alpha_1,\alpha_2,\alpha_\Delta\in\KK_\infty$ and $\gamma\in\KK$ such 
that
$\alpha_1(\|x\|)\leq \VV(x) \leq \alpha_2(\|x\|)$, $\VV(f(x))-\VV(x)\leq 
-\alpha_\Delta(\|x\|)+\gamma(\|w\|)$ for 
all $x\in\SS$, $w\in\WW$   
is an input-to-state stable (ISS) Lyapunov function  for  
$f$ in $\SS$ with respect to $w$.
\end{definition}

\begin{result}
 Given $x(t+1)=f(x(t))$, $x\in\rr^n$, and a PI $\SS$ for $f$, $0\in\SS$, if there exists 
a Lyapunov function for $f$ in $\SS$, the origin 
is asymptotically 
stable (AS) for  $f$ with domain of attraction $\SS$.
Given $x(t+1)=f(x(t),w(t))$, $x\in\rr^n$, $w\in\WW\subseteq \rr^d$, and a RPI $\SS$ for 
$f$, 
$0\in\SS$, if there exists 
a ISS Lyapunov function for $f$ in $\SS$, the origin 
is ISS for  $f$  with respect to $w$ with domain of attraction $\SS$.
\end{result}

We consider the  uncertain constrained discrete-time systems with sampling period $T_s$,
 \bsube \label{eq:system}
 \beqar 
   x(t+1) &=& \sum_{i=1}^{\ell }[\bar \xi]_iA_ix(t)+Bu(t),\label{eq:pLDI} \\
            && x\in\XX,\ u\in\UU \label{eq:sysConstr} 
 \eeqar
 \esube
where 
$A_i\in\rr^{n\times n}$, 
$i\in \zz_{[1,\ell]}$ and $B$ are known matrices of appropriate size, and $\XX\subseteq \rr^n$, 
$\UU\subseteq \rr^m$ 
are constraints on system states and inputs. In~\eqref{eq:system}, the uncertainty is associated 
to 
$\bar \xi \in 
\Xi\subset  
\rr^\ell$, which is unknown and constant or changing slowly with respect to the system dynamics, 
and $\Xi=\{\xi\in\rr^\ell: 
0\leq [\xi]_i\leq 1$, $\sum_{i=1}^{\ell }[\xi]_i=1\}$. We call $\bar \xi$  {\em convex 
combination 
vector}, since it 
describes a convex combination of the ``vertex systems'' $f_i(x,u)=A_ix+Bu$, $i\in 
\zz_{[1,\ell]}$.
\begin{assumption}\label{ass:estimate}
An estimator is computing the estimate $ \xi(t)$ of $\bar \xi$ such that $ \xi(t)\in\Xi$ for all 
$t\in\zz_{0+}$.
\end{assumption}
We denote by $\tilde \xi(t) = \bar \xi(t)- 
\xi$,  the 
estimation error at time $t$ for which it holds that $\tilde \xi(t) +\xi \in \Xi$. Given $\xi\in\Xi$, for 
shortness we 
write $\tilde \xi \in \tilde \Xi(\xi)$,  where $\tilde 
\Xi{(\xi)} = \{\tilde \xi\in\rr^\ell :\ \exists\bar \xi 
\in\Xi,\ \tilde\xi  =\bar \xi - \xi\}$ is the set of 
possible estimation error vectors. Assumption~\ref{ass:estimate} is what is required from the 
estimator for the development in this paper to hold. Some comments on how to design estimators 
that satisfy Assumption~\ref{ass:estimate} are given later, in Section~\ref{ssec:Estimator}
\brem
The trajectories produced by~\eqref{eq:pLDI} are a subset of those of the
 pLDI
 \beq\label{eq:pLDI1}
 x(t+1)\in\co\{A_ix(t)+Bu(t)\}_{i=1}^\ell. 
\eeq 
The pLDI~\eqref{eq:pLDI1} is equivalent to~\eqref{eq:pLDI} if a varying parameter vector, i.e.,  
$\bar \xi(t)$, is considered.
\erem

Consider the finite time optimal control problem
     \bsube\label{eq:FHOCparam}
          \beqar
            \VV_{\xi_t^N}^\MPC (x(t)) =  & \hspace*{-0.25cm} \min\limits_{U_t }& 
            x_{N|t}'\PP(\xi_{N|t})x_{N|t}+\\ &&\qquad \sum_{k=0}^{N-1} 
x_{k|t}'Qx_{k|t}+u_{k|t}'Ru_{k|t}\quad \label{eq:FHOCparam-cost}\\
            & \hspace*{-0.15cm} {\rm s.t.}& \hspace*{-0.25cm}   x_{k+1|t}=\sum_{i=1}^{\ell 
}[  \xi_{k|t}]_iA_ix_{k|t}+Bu_{k|t} \quad \quad\label{eq:FHOCparam-dyn}\\
                             && \qquad u_{k|t}\in\UU,\ x_{k|t}\in\XX \label{eq:FHOCparam-constr} \\ 
                             &&\qquad (x_{k|t},u_{k|t})\in \CC_{xu} \label{eq:FHOCparam-RCI}\\
                             &&\qquad x_{N|t}\in \XX_N \label{eq:FHOCparam-termSet}\\
                             &&\qquad x_{0|t} = x(t),\label{eq:FHOCparam-init}
                            \eeqar
                            \esube 
where $N\in\rr_+$ is the prediction horizon, $Q\in\rr^{n\times n}$, $R\in\rr^{m\times m}$, 
$Q,R>0$, 
$\PP(\xi)\in\rr^{n\times n}$, $\PP(\xi)>0$, for all $\xi \in \Xi$, $\CC_{x,u}\subseteq 
\XX\times\UU$,  $U_t 
=[u_{0|t}\ 
\ldots u_{N-1|t}]$ is the sequence of control inputs along the prediction horizon, and $\xi_t^N 
=[\xi_{0|t}\ 
\ldots \xi_{N|t}]\in\Xi^{N+1}$ is a sequence of predicted parameters,  not 
necessarily
constant. Let $U_t^* 
=[u_{0|t}^*\ 
\ldots u_{N-1|t}^*]$ be the solution of~\eqref{eq:FHOCparam} at $t\in\zz_{0+}$.

\begin{problem}\label{prob:ampc}
Given~\eqref{eq:system} and an estimator producing the sequence of estimates 
$\{\xi(t)\}_t$ 
such that $\xi(t)\in\Xi$ for all $t\in\zz_{0+}$ according to Assumption~\ref{ass:estimate}, design 
the  
sequence of  predicted convex combination vectors $\xi_t^N $, the terminal cost 
$\PP(\xi)$, 
the robust terminal set 
$\XX_N$, and  the robust constraint set $\CC_{xu}$ in~\eqref{eq:FHOCparam} so that the 
IAMPC 
controller 
that at 
any $t\in\zz_{0+}$ solves~\eqref{eq:FHOCparam} and applies $u(t)=u_{0|t}^*$ 
achieves: $(i)$ 
ISS of the 
closed-loop 
with respect to $\tilde \xi_{0|t} =\bar \xi -  \xi_{0|t}$, $(ii)$ robust satisfaction of the 
constraints including when $\tilde \xi_{0|t} \neq 0$, $(iii)$  
guaranteed 
convergence of the  runtime numerical algorithms and computational load comparable 
to a (non-adaptive) MPC.
\end{problem}

In Problem~\ref{prob:ampc}, $(i)$ is concerned with conditions on the 
behavior during the estimation 
transient and when the estimate has converged, $(ii)$ is concerned with the system safety in 
terms of constraints satisfaction, and $(iii)$ is concerned with computational requirements, 
especially due to recent 
applications to fast 
systems~\cite{MPCautoSurvey12,dicairano12}. 

Problem~\ref{prob:ampc} requires robust constraint satisfaction (as 
in~\cite{KBM96,cuzzola02,tubeMPC}) and  ISS, i.e.,  a proportional 
effect of the estimation error on the closed-loop Lyapunov function.  
 The rationale for seeking ISS rather than robust stability as 
 in~\cite{KBM96,cuzzola02} is that,  when  the unknown parameters  do not change or 
change 
 slowly, a ``well designed''
estimator should 
eventually converge to correct value, and hence, by ISS definition,  the 
closed-loop becomes AS. However, ISS holds regardless of the estimator 
convergence as long as Assumption~\ref{ass:estimate} is satisfied, as well as
robust constraint satisfaction, which has to be guaranteed also in presence of estimation error and 
thus is guaranteed even if the parameters 
change. While this paper focuses on a control design independent from the estimator design,  the 
dependency of the closed-loop performance on estimation error is captured by the expansion term 
in the ISS Lyapunov function.

Consider the linear parameter-varying (LPV) system
\beq\label{eq:LPV} 
x(t+1)=\sum_{i=1}^\ell [\xi(t)]_iA_ix(t)+Bu(t),
\eeq 
where for all $t\in\zz_+$, $\xi(t)\in\Xi$,  the parameter-dependent (linear) 
control law
\beq\label{eq:parTermCtrl}
  u=\kappa(\xi)x= \left( \sum_{i=1}^\ell [\xi]_iK_i \right) x,
\eeq
and the parameter-dependent (quadratic)   function  
\beq\label{eq:parLF} 
\VV_\xi(x) = x'\PP(\xi)x = x'\left( \sum_{i=1}^\ell [\xi]_iP_i \right) x ,
\eeq
where $P_i>0$, $i\in \zz_{[1,\ell]}$.

\begin{definition}[\cite{daafouz01}]\label{res:pLFcomp}
A function~\eqref{eq:parLF} 
such 
that  
$\VV_{\xi(t+1)}(x(t+1))-\VV_{\xi(t)}(x(t))\leq 0$,
 for all $\xi(t),\xi(t+1)\in\Xi$,
where equality holds only if $x=0$, is a {\em parameter-dependent Lyapunov function} 
for~\eqref{eq:LPV} in closed-loop with \eqref{eq:parTermCtrl}.  \cvd
\end{definition}

By~\cite{daafouz01,cuzzola02,dic15}, given $Q\in\rr^{n\times n}$, $R\in\rr^{m\times m}$
$Q,R>0$, any solution $G_i,S_i\in\rr^{n\times n}$, $S_i>0$,  $E_i\in\rr^{m\times n}$, $i\in 
\zz_{[1,\ell]}$, of
\algn{\label{eq:pLF-designLMI}
 \bsm
   G_i + G_i' - S_i   	& 	(A_iG_i+BE_i)' & E_i' 		& G_i' \\
   (A_iG_i+BE_i) 	  & S_j 			& 0 		 & 0 \\
   E _i& 0 & R^{-1} & 0 \\ 
   G_i & 0 &0 & Q^{-1}
   \esm > 0,\ \forall i,j\in \zz_{[1,\ell]}.
   }
is such that \eqref{eq:parTermCtrl}, \eqref{eq:parLF}  where $P_i=S_i^{-1}$, 
$K_i=E_iG_i^{-1}$,  $i\in \zz_{[1,\ell]}$,
satisfy 
\algn{ &\VV(x(t+1),\xi(t+1))-\VV(x(t),\xi(t))
 \leq \label{eq:pLFdef} \\ &-x(t)'(Q+\kappa(\xi(t))'R\kappa(\xi(t)))x(t)',\ 
 \forall \xi(t),\xi(t+1)\in\Xi \no}
  for the closed-loop~\eqref{eq:LPV}, \eqref{eq:parTermCtrl}. 
 \begin{assumption}\label{ass:LMIfeasible}
 For the given  $A_i$, $i\in\zint{1,\ell}$, $B$, $Q$, $R$, \eqref{eq:pLF-designLMI} admits a 
feasible solution
 \end{assumption}
 The LMI~\eqref{eq:pLF-designLMI} is a relaxation of those in~\cite{KBM96,cuzzola02} since it 
allows for a  parameter-dependent Lyapunov function and a parameter-dependent linear control 
law. Thus, Assumption~\ref{ass:LMIfeasible} is  more relaxed of and implied by the existence of 
an 
(unconstrained) stabilizing linear control law for the uncertain system~\eqref{eq:pLDI} , see, 
e.g.,~\cite{KBM96,cuzzola02}. Indeed, if the vertex systems are such that the
uncertainty is too large, \eqref{eq:pLF-designLMI}~may be infeasible, similarly to the case where 
a stabilizing controller for an uncertain system does not exist. However, 
since~\eqref{eq:pLF-designLMI} is used here for design, such situation will be recognized before 
controller execution and corrective measures, such as improving the engineering of the plant or 
resorting to other control techniques can be actuated. By
using~\eqref{eq:pLF-designLMI} only for design, as opposed to~\cite{KBM96,cuzzola02}, the 
proposed method solves online only QPs, which makes it feasible also for  
applications with fast dynamics and low-cost microcontrollers~\cite{dicairano12,MPCautoSurvey12}.
\brem
Here we 
consider the case where $B$ in~\eqref{eq:LPV} is independent of the uncertain parameters due to 
the limited space, 
as this allows to shorten several derivations. The expanded derivations related to the case of 
uncertain $B$ will be included in future/extended versions of this work.
\erem
\sdcRemoved{In~\cite{dic15}, by solving~\eqref{eq:pLF-designLMI}, 
 $\PP(\xi)$ and $\XX_N$ in~\eqref{eq:FHOCparam} are  designed such that for all constant 
and 
known values 
$\xi\in \Xi$, when $\tilde 
\xi=0$   the closed-loop is AS. Thus, the MPC controller can be deployed before the exact 
parameter values are known, and, after  open-loop identification is done, the controller is 
adjusted  without the need of  re-design algorithms. Here, we build upon   this method to 
solve 
Problem~\ref{prob:ampc}, where the parameter estimation and the MPC operates 
concurrently, 
thus obtaining an indirect-adaptive MPC strategy.}

%
%
%
%
%

\section{Unconstrained IAMPC: ISS Property}\label{sec:unconstr}
We start from the unconstrained case, $\XX= \rr^n$, $\UU= 
\rr^m$. 

\subsection{Stability with parameter prediction along the horizon}

We begin by considering a simpler case where $\xi_{k|t}= \bar \xi(t+k)$, 
$k\in\zint{0,N}$, where it is possible that $\xi_{k_1|t}\neq \xi_{k_2|t}$, for 
$k_1,k_2\in\zint{0,N}$. This corresponds to controlling an LPV system with preview on the 
parameters for $N$ steps in the future, but no information afterwards.

\begin{lemma}\label{lem:LPVstable}
Let Assumption~\ref{ass:LMIfeasible} hold and consider~\eqref{eq:LPV} and the MPC that at 
$t\in\zz_{0+}$ 
solves~\eqref{eq:FHOCparam} where $\XX_N=\XX= \rr^n$, $\UU= \rr^m$, 
$\CC_{xu}= 
\rr^{n+m}$, $\UU= \rr^m$, $  \xi_{k|t}= \bar \xi(t+k)$, and $\PP(\xi)$, $\kappa(\xi)$ are 
from~\eqref{eq:pLF-designLMI}. Then, the origin is AS for the closed loop with domain of 
attraction 
$\rr^n$ for 
every sequence $\{\bar \xi(t)\}_{t}$, such that $\bar \xi(t)\in\Xi$, for all $t\in\rr_{0+}$.
\end{lemma}

\begin{proof}
We follow the proofs for unconstrained MPC extended to time-varying systems, 
see,~\cite[Sec.2.4]{RM09}. 
For~\eqref{eq:LPV} in closed-loop with~\eqref{eq:parTermCtrl} designed 
by~\eqref{eq:pLF-designLMI}, \eqref{eq:pLFdef} holds, where $\VV$ in~\eqref{eq:parLF} is 
designed also 
by~\eqref{eq:pLF-designLMI}. 
Thus, in~\eqref{eq:FHOCparam}, $\VV_{\xi^N}^\MPC(x)$, is
lower and 
upper 
bounded by 
class $\KK_\infty$ functions, $\unl\alpha(\|x(t)\|)=\lambda_\rmin(Q)\|x(t)\|^2$, 
$\ovl\alpha(\|x(t)\|)=\psi \upsilon \|x(t)\|^2$, where $\psi\in\rr_+$ and $\upsilon =
\max_{i\in\zint{1,\ell}}\lambda_\rmax{(P_i)}$, for any $\xi^N_t\in\Xi^{N+1}$, 
see~\cite[Sec. 2.4.5]{RM09}. From $x(t)$,  let 
$U^*(t)=[u^*_{0|t}\ldots u^*_{N-1|t}]$ 
be the optimal solution 
of~\eqref{eq:FHOCparam}. 
At 
$t+1$ from 
$x(t+1)=x_{1|t}$,   $[ \tilde u_{0|t+1}\ldots \tilde u_{N-1|t+1}]$ where 
$\tilde  u_{k|t+1} 
= 
u^*_{k+1|t}$ 
for $k\in\zz_{[0,N-2]}$, $ \tilde u_{N-1|t+1}=\sum_{i=1}^\ell  [\xi_{N-1|t+1}]_i K_i$,  
has cost 
$\tilde J\leq 
\VV_{\xi_{t}^N}^\MPC(x(t))-x(t)'Qx(t)$, due to~\eqref{eq:pLFdef}  and  
$\xi_{k|t+1}=\xi_{k+1|t}$, for all $k\in\zint{0,N-1}$. Since 
$\VV_{\xi_{t+1}^N}^\MPC(x(t+1))\leq 
\tilde J$, 
\algn{ \no 
&\VV_{\xi_{t+1}^N}^\MPC(x(t+1))-\VV_{\xi_{t }^N}^\MPC(x(t)) \leq 
-x(t)Qx(t)\\ \no & \qquad \qquad \leq - 
\lambda_\rmin(Q)\|x(t)\|^2 = 
\alpha_\Delta(\|x(t)\|),} and $\alpha_\Delta\in\KK_\infty$. Thus, 
$\VV_{\xi_t^N}^\MPC(x(t))$ is a Lyapunov 
function for the closed-loop system for any $\xi_t^N$ such that $\xi_{t|k}=\bar \xi(t+k)$, 
and hence the origin is AS in $\rr^n$.
\end{proof}

By Lemma~\ref{lem:LPVstable}, the MPC based on~\eqref{eq:FHOCparam} with  perfect preview 
along the horizon is
stabilizing. 
Next we  account for the effect of 
the parameter estimation error.

\subsection{ISS with respect to parameter estimation error}

Consider now the case relevant to Problem~\ref{prob:ampc} where $\bar \xi(t)$ is 
constant, i.e., 
$\bar \xi(t)=\bar \xi$, for all $t\in\zz_{0+}$, unknown, and being estimated. Thus,  
$\tilde \xi_{0|t}=\bar \xi - \xi_{0|t}$ is the error in the parameter 
estimate, which may be time-varying, and  $\tilde \xi_{0|t}\in\Xi{(\xi_{0|t})}$. The parameter 
estimation error induces a state prediction error 
 \algn{ \label{eq:pLDIerr} 
   \eps_x &= \sum_{i=1}^{\ell }[\bar \xi]_iA_ix -\sum_{i=1}^{\ell }[ \xi_{0|t}]_iA_ix = 
\sum_{i=1}^{\ell 
}[\tilde \xi_{0|t}]_iA_ix .}
Indeed,
\algn{ \no \|\eps_x\| &=  \left\|\sum_{i=1}^{\ell }[\tilde \xi_{0|t}]_iA_ix\right\| \leq  
\left\|\sum_{i=1}^{\ell 
}[\tilde 
\xi_{0|t}]_iA_i\right\|\cdot  \left\|x\right\| \\ &  \leq \left(  \sum_{i=1}^{\ell 
} | [\tilde 
\xi_{0|t}]_i|\left\|A_i\right\|\right)  \left\|x\right\|\leq \gamma_A \|\tilde \xi_{0|t}\|_1 \|x\|}
where $\gamma_A = \max_{i=1,\ldots \ell }\|A_i\|$. 

Consider  the value function $\VV^\MPC_{\xi^N}$ of~\eqref{eq:FHOCparam}, the following result 
is straightforward from~\cite{RM09}.  
\begin{result}\label{res:LipCont}
For every compact $\XX_L\subseteq \rr^n$, the value 
function of~\eqref{eq:FHOCparam}, where $\PP(\xi)$ is designed according 
to~\eqref{eq:pLF-designLMI}, is 
Lipschitz-continuous in $x\in\XX_L$, that is, there exists $L\in\rr_+$ such that for every $  
x_1,x_2\in\XX_L$, 
$\|\VV_{\xi^N}^\MPC(x_1)-\VV_{\xi^N}^\MPC(x_2)\|\leq L \|x_1-x_2\|$,
for every ${\xi^N}\in\Xi^{N+1}$.
\end{result}
Result~\ref{res:LipCont} follows directly from the fact that for every 
${\xi^N}\in\Xi^{N+1}$, $\VV_{\xi^N}^\MPC$ is
piecewise quadratic~\cite{RM09} and hence it is Lipschitz continuous in any compact set $\XX_L$. 
Thus, for any 
$\XX_L\subseteq \rr^n$ and ${\xi^N}\in\Xi^{N+1}$, there exists a Lipschitz parameter 
$L_{\xi^N}\in\rr_+$.  Since $\Xi^{N+1}$ is compact, i.e., closed and bounded, there exists a 
maximum  of $L_{\xi^N}\in\rr_+$ for  $\xi^N\in\Xi^{N+1}$. Such maximum is the Lipschitz 
constant $L$.

\begin{lemma}\label{lem:ISSxw}
Let $\xi_{k-1|t+1}=\xi_{k|t}$, for all $k\in\zint{1,N}$, $t\in\zz_{0+}$. 
Then, there exists $\gamma_L>0$, such that for 
every $x\in\XX_L$, 
 \algn{ \label{eq:ISSxw}
\VV_{\xi_{t+1}^N}^\MPC(x(t+1)) \leq & \VV_{\xi_t^N}^\MPC(x(t)) -\lambda_\rmin(Q) 
\|x(t)\|^2 \no \\ &  \quad +  
\gamma_L \|\tilde 
\xi_{0|t}\|_1 \|x(t)\|.
 }
\end{lemma}

\begin{proof}
By Lipschitz continuity of $\VV_{\xi_t^N}^\MPC(x)$, 
$$\VV_{\xi_{t+1}^N}^\MPC(x(t+1))-\VV_{\xi_{t+1}^N}^\MPC(  x_{1|t})\leq  L\|\eps_x(t)\|\leq 
\gamma_A L\|\tilde\xi_{0|t}\|_1 
\|x\|$$
Also, due to the result of Lemma~\ref{lem:LPVstable},  
$$
\VV_{\xi_{t+1}^N}^\MPC(x_{1|t})-\VV_{\xi_{t}^N}^\MPC(x(t))\leq -\lambda_\rmin(Q) \|x\|^2
$$
Thus,
$$\VV_{\xi_{t+1}^N}^\MPC(x(t+1))\leq \VV_{\xi_{t}^N}^\MPC(x(t)) -\lambda_\rmin(Q) 
\|x\|^2 +   \gamma_L \|\tilde\xi_{0|t}\|_1 \|x\|  $$
where $\gamma_L=\gamma_A L$. 
\end{proof}

\begin{theorem}\label{th:ISS}
Let
Assumptions~\ref{ass:estimate}, \ref{ass:LMIfeasible} hold, and $\xi_{k|t+1}=\xi_{k+1|t}$, 
for 
all $k\in\zint{0,N-1}$, and all $t\in\zz_{0+}$. 
For the MPC that at any step solves~\eqref{eq:FHOCparam}  where $\PP(\xi)$ is 
designed 
according 
to~\eqref{eq:pLF-designLMI},  $\XX_N=\XX= \rr^n$, $\UU= \rr^m$, $\CC_{xu}= 
\rr^{n+m}$, $\UU= \rr^m$,  $\VV^\MPC_{\xi^N}(x)$ is an ISS-Lyapunov 
function with 
respect to the estimation error 
$\tilde \xi_{0|t}=\bar \xi-\xi_{0|t}\in \tilde \Xi(\xi_{0|t})$  for~\eqref{eq:system} in 
closed loop 
with the MPC 
based 
on~\eqref{eq:FHOCparam} in any  $\XX_\eta\subseteq \XX_L$, where $\XX_\eta$ is RPI with 
respect 
to 
$\tilde \xi_{0|t}$ for the closed loop.
\end{theorem}
\begin{proof}
By Lemma~\ref{lem:ISSxw} we have that
$$\VV_{\xi_{t+1}^N}^\MPC(x(t+1))\leq \VV_{\xi_{t }^N}^\MPC(x(t)) -\lambda_\rmin(Q) 
\|x\|^2 +   \gamma_L \|\tilde \xi_{0|t}\|_1 \|x\| .$$

Due to the norm equivalence in finite 
dimensional spaces, for a $p$-norm, there exists $\gamma_p$ 
such that $\|\tilde \xi_{0|t}\|_1\leq \gamma_p\|\tilde \xi_{0|t}\|$ for every $\tilde 
\xi_{0|t}\in\rr^n$. Hence, 
$$\VV_{\xi_{t+1}^N}^\MPC(x(t+1))\leq \VV_{\xi_{t }^N}^\MPC(x(t)) -\lambda_\rmin(Q) 
\|x\|^2 +   \gamma_L\gamma_p \|\tilde \xi_{0|t}\| \|x\| .$$

Since $\XX_\eta\subseteq \XX_L$ and $\XX_L$ is compact,
there exists a finite $\gamma_\eta>0 $ such that $\|x\|\leq \gamma_\eta$ for all $x\in\XX_\eta$. 
Then, for 
all $x\in\XX_\eta$, where $\XX_\eta$ is RPI with respect to $\tilde \xi_{0|t}\in \tilde 
\Xi(\xi_{0|t})$,
\beq\label{eq:issLF} \VV_{\xi_{t+1}^N}^\MPC(x(t+1))\leq \VV_{\xi_{t }^N}^\MPC(x(t)) 
-\lambda_\rmin(Q) 
\|x(t)\|^2 +   \gamma_{\rm ISS} \|\tilde \xi_{0|t}\|\eeq
and the closed-loop is ISS with respect to $\tilde  \xi_{0|t}$, with ISS Lyapunov 
function 
$\VV_{\xi_{t}^N}^{\rm MPC}$ and ISS gain $\gamma_{\rm 
ISS}=\gamma_L\gamma_p\gamma_\eta$.
\end{proof}

 Theorem~\ref{th:ISS}  requires the existence of a RPI set $\XX_\eta$ contained in $\XX_L$ 
because that 
is where the value function is Lipschitz-continuous. For $\eta>0$,  
$\XX_\eta\subseteq \XX_L$ such that for all $\xi^N\in\Xi^{N+1}$, $\VV_{\xi^N}^\MPC\leq 
\eta$ 
and 
for all 
$x\in\XX_\eta$, $\VV_{\xi^N}^\MPC-\lambda_\rmin(Q)\|x\|^2\leq \eta - 
2\gamma_L\|x\|$ is RPI, 
because of Lemma~\ref{lem:ISSxw} and 
$\|\tilde \xi\|_1\leq 2$, for all $\tilde \xi\in\tilde \Xi(\xi)$. For the case of 
constrained IAMPC where $\XX$ is compact, for ensuring constraint satisfaction we need to 
construct a compact RPI set 
$\XX_\eta\subseteq\XX$, so that we can define $\XX_L= \XX_\eta$, and 
Result~\ref{res:LipCont} holds in $\XX_\eta$. 
Next, we focus on the case of  $\XX$, $\UU$  
compact and how to build $\XX_\eta$  for constrained IAMPC. 
\sdcRemoved{
\brem A proof for Theorem~\ref{th:ISS} can be obtained assuming {\em 
uniform continuity}, instead of  
Assumption~\ref{ass:LipCont}, following~\cite[Th.2]{LimonISS}, where it is pointed 
out~\cite[A.4.3]{LimonISS} that the two assumptions are equivalent 
for~\eqref{eq:FHOCparam}, given $\xi_t^N\in\Xi_t^{N+1}$.
\erem
}

\section{Constrained IAMPC: Robust Constraints}\label{sec:constr}
By designing the terminal cost from a pLF as in Section~\ref{sec:unconstr}, the 
closed loop of~\eqref{eq:system} with the unconstrained IAMPC that 
solves~\eqref{eq:FHOCparam} 
is ISS with 
respect to $\tilde \xi_{0|t}$. 
Next, we consider constrained IAMPC, i.e., $ \XX\times\UU\subset \rr^n\times \rr^m$.
\begin{assumption}\label{ass:compact}
$\XX$, $\UU$ are compact sets, $0\in\intset(\XX)$, $0\in\intset(\UU)$.
\end{assumption} 

Under Assumption~\ref{ass:compact}, we first show that for the LPV 
system~\eqref{eq:LPV} with perfect preview along the prediction horizon, i.e.,  $\xi_{k|t} 
=\bar \xi(t+k)$, for all $k\in\zint{0,N}$, the closed-loop recursively satisfies the 
constraints and is AS. 
Then, we enforce constraint satisfaction  when $\tilde \xi_{0|t}\neq 
0$. These objectives are achieved by designing~$\XX_N$ and $\CC_{xu}$ 
in~\eqref{eq:FHOCparam}, 
respectively.

\subsection{Terminal set design for nominal terminal constraint}

Consider~\eqref{eq:LPV} where $\xi(t)$ is known at  $t\in\zz_{0+}$ 
and the control law~\eqref{eq:parTermCtrl} resulting in the
closed-loop LPV system 
\beq\label{eq:LPVcl} 
x(t+1)=\sum_{i=1}^\ell [\xi(t)]_i(A_i+BK_i)x(t).
\eeq 
The trajectories of~\eqref{eq:LPVcl} are contained in those of the pLDI
\beq\label{eq:pLDIcl} 
x(t+1)\in \co\{(A_i+BK_i)x(t)\}_{i=1}^\ell.
\eeq 

For \eqref{eq:pLDIcl} in closed loop with~\eqref{eq:parTermCtrl} designed 
by~\eqref{eq:pLF-designLMI}
subject to~\eqref{eq:sysConstr},
in~\cite{dic15} it was shown that the maximum 
constraint admissible set $\XX^\infty\subseteq \bar 
\XX$, where $\bar \XX=\{x\in\XX:\ 
\kappa(\xi)x\in\UU, \forall \xi\in\Xi\}$ 
is 
polyhedral, finitely determined and  has non-empty interior with $0\in 
\intset(\XX^\infty)$. 
 $\XX^\infty$ is RPI  for~\eqref{eq:LPVcl} for all $\xi\in\Xi$, and is  the limit of a sequence of 
backward reachable sets. Let 
$\XX_{x u}$ be a given set of feasible states and inputs  $\XX_{x u}\subseteq 
\XX\times\UU$, $0\in\intset(\XX_{xu})$ and let 
\algn{\XX^{(0)} &= \{x:\ (x,K_ix)\in\XX_{xu},\ \forall i\in 
\zz_{[1,\ell]}\}\no \\
\XX^{(h+1)} &= \{x:\ 
(A_i+BK_i)x\in\XX^{(h)},\ 
\forall 
i\in \zz_{[1,\ell]}\}\cap \XX^{(h)}\nonumber \\
\XX^{\infty} &= \lim_{h\rar\infty}\XX^{(h)}. \label{eq:MCAS}
}
Due to the finite 
determination of $\XX^\infty$    there exists a finite $\bar h\in\zz_{0+}$ 
such that 
$\XX^{(\bar h+1)}=\XX^{(\bar 
h)}=\XX^\infty$, i.e., the limit in~\eqref{eq:MCAS} is reached in a finite 
number of iterations.

\begin{lemma}\label{lem:ConstrLPVstable}
Consider~\eqref{eq:LPV} and the MPC that at $t\in\zz_{0+}$ 
solves~\eqref{eq:FHOCparam} where $\XX\subset \rr^n$, $\UU\subset \rr^m$, $\CC_{xu}= 
\rr^{n+m}$, $\xi_{k|t}= \bar \xi(t+k)$, $\PP(\xi)$, $\kappa(\xi)$ are designed according 
to~\eqref{eq:pLF-designLMI} and $\XX_N=\XX^{\infty}$, where $\XX^\infty$ is 
from~\eqref{eq:MCAS}. At a given $t\in\zz_{0+}$, let $x({ t})\in\XX$, 
$\xi_t^N \in\Xi^{N+1}$ be  such that \eqref{eq:FHOCparam} is feasible. Then, 
\eqref{eq:FHOCparam} is feasible for any $\tau \geq  t$, i.e., $\XX_f(\xi^N)=\{x\in\XX: 
\eqref{eq:FHOCparam}\ {\rm feasible\ for\ } x_{0|t}=x,\ 
\xi_{k|t}=\xi_k\in\Xi,\ k\in\zint{0,N}\}$ is a PI set, and the 
origin is AS in $\XX_f(\xi^N)$.
\end{lemma}

\begin{proof}
First we show that  
$\XX_f(\xi^N)\subseteq \XX$ is PI for the closed-loop system, i.e., if 
$x(t)\in\XX_f(\xi_t^N)$, then
$x(t+1)\in\XX_f(\xi_{t+1}^N)$, for all $\xi(t+1+N)\in\Xi$.
 Since 
$x(t)\in\XX_f(\xi_t^N)$,
there 
exists  $U^*(t)=[u^*_{0|t} \ldots u^*_{N-1|t}]$ optimal (and
feasible) for~\eqref{eq:FHOCparam}. From 
$x(t+1)=x_{1|t}$, given $\xi_{t+1}^N$,  consider  $\tilde U = [\tilde 
u_{0|t+1}\ldots \tilde 
u_{N-1|t+1}]$, where 
$\tilde  u_{i|t+1} 
= 
u^*_{i+1|t}$ 
for $i\in\zz_{[0,N-2]}$, $ \tilde u_{N-1|t+1}=\kappa(\xi_{N|t}) x_{N|t}$.  By $ 
\xi_{t|k}=\bar \xi({t+k})$, it holds  
$\xi_{k|t+1}=\xi_{k+1|t}$, and  
the trajectory generated by $\tilde U$ is such that $\tilde x_{k|t+1}=x_{k+1|t}$, $\tilde 
x_{k|t+1}\in\XX$, $\tilde u_{k|t+1}\in\UU$  for all $k\in\zz_{[0,N-1]}$. Since 
$\XX_N=\XX^\infty$,
$x_{N-1|t+1}=x_{N|t}\in\XX^\infty\subseteq \XX$ and  $\kappa(\xi_{N-1|t+1}) x_{N|t}\in\UU$
for 
all 
$\xi_{N-1|t+1}\in\Xi$, hence~\eqref{eq:FHOCparam-constr} is satisfied. 
Also, $x_{N|t+1}\in\XX^N$, because  $\XX^N =\XX^\infty$ is PI for~\eqref{eq:LPV} 
in closed loop with~\eqref{eq:parTermCtrl} for all $\xi\in \Xi$.  
Thus, also~\eqref{eq:FHOCparam-termSet} is satisfied and $\tilde U$ is feasible 
from 
$x(t+1)$ for any $\xi_{t+1}^N$ that is  admissible according to the assumptions, and 
$x(t+1)\in\XX_f(\xi_{t+1}^N)$. Hence,   
$\XX_f(\xi^N)$ is 
PI. AS
follows by the same arguments of  
Lemma~\ref{lem:LPVstable} with $\VV_{ \xi^N }^\MPC$ as Lyapunov function, and with
$\XX_f(\xi^N)$ as domain of attraction.
\end{proof}

Next, we  ensure  robust  satisfaction of~\eqref{eq:FHOCparam-constr}, 
\eqref{eq:FHOCparam-termSet}, 
in presence of estimation error $\tilde \xi_{0|t}\neq 0$.

\subsection{Robust constraints design}

In order to ensure robust constraint satisfaction in the presence of parameter estimation error we 
design the constraint~\eqref{eq:FHOCparam-RCI} from a RCI set for the 
pLDI~\eqref{eq:pLDI1}, whose trajectories include those of~\eqref{eq:pLDI}.
Based on Definition~\ref{def:RCI}, let $\CC\subseteq \XX$ be a convex set such that for 
any $x\in\CC$ there exists $u\in\UU$ such that $A_ix+Bu\in\CC$ for all $i\in\zint{1,\ell}$. 
Given $\CC$, we design  $\CC_{xu}$ in \eqref{eq:FHOCparam-RCI} as 
\beq\label{eq:cxu}
 \CC_{xu}=\{(x,u)\in\CC\times\UU,\ A_ix+Bu\in\CC,\ \forall i\in\zint{1,\ell}\},
 \eeq 
that is, the state-input pairs that 
result in states within the RCI set for any vertex system of the 
pLDI~\eqref{eq:pLDI1}.

\begin{lemma}\label{lem:RCI}
Consider~\eqref{eq:FHOCparam} where $\XX_N=\rr^n$, and $\CC_{xu}$ 
in~\eqref{eq:FHOCparam-RCI} is defined 
by~\eqref{eq:cxu}.  If $x(t)\in\CC$, \eqref{eq:FHOCparam} is 
feasible 
for all $\tau \geq t$, for any $\xi_\tau^N\in\Xi^{N+1}$ and any $\tilde \xi_{0|\tau} \in\tilde 
\Xi(\xi_{0|\tau})$. 
\end{lemma}
\begin{proof}
Due to the definition of $\CC_{xu}$, for all $x\in\CC$ there exists $u\in\UU$ such that 
$(x,u)\in\CC_{xu}$. Thus, if $(x(t),u(t))\in\CC_{x u}$,  by convexity  
$x(t+1)=\sum_{i=1}^\ell [\bar \xi]_i
A_ix(t)+Bu(t)\in\CC\subseteq \XX$, for all $\xi_{0|t}+\tilde \xi_{0|t}=\bar \xi\in\Xi$,  i.e., 
\eqref{eq:FHOCparam-RCI} and hence~\eqref{eq:FHOCparam-constr}  are 
satisfied for 
any actual $\bar \xi\in\Xi$. Since $x(t+1)\in\CC$, 
the reasoning can be repeated proving robust feasibility for any $\tau \geq t$.
\end{proof}


$\CC$ can be computed as the maximal RCI set for~\eqref{eq:pLDI1} 
from the sequence~\cite{BM08},
 \bsube \label{eq:mRCIcompute}
\algn{
\CC^{(0)}&= \XX,\\ 
\CC^{(h+1)}&= \{x:\ \exists u\in\UU, \nonumber
\\ &\hspace{0.8cm}A_ix+Bu\in\CC^{(h)},\  \forall 
i\in\zint{1,\ell}\}\cap \CC^{(h)}.
 }\esube 
The maximal RCI set in $\XX$  is the fixpoint of~\eqref{eq:mRCIcompute}, i.e., 
$\CC^\infty =\CC^{(\bar h)}$ such that $\CC^{(\bar h+1)}=\CC^{(\bar h)}$, and  is the largest 
set  within $\XX$ that can be made invariant  
for~\eqref{eq:pLDI1} with inputs in $\UU$.

Based on Lemma~\ref{lem:RCI}, by  $\CC_{xu}$ in~\eqref{eq:cxu} we obtain 
constraint 
satisfaction even when $\tilde \xi_{0|t}\neq 0$. However, the 
maximal RCI 
set does not guarantee that the terminal constraint can be 
satisfied, that is, it may not be possible to reach $\XX_N$ in $N\in\zz_+$ steps for all 
$x\in\CC$  
by trajectories such that $(x_{k|t},u_{k|t})\in\CC_{xu}$. Furthermore, 
for Lemma~\ref{lem:ConstrLPVstable} to hold,  the control inputs generated 
by~\eqref{eq:parTermCtrl} that make $\XX_N$  
PI for~\eqref{eq:LPVcl} must be feasible for~\eqref{eq:FHOCparam}, that is, 
$(x,\kappa(\xi)x)\in\CC_{xu}$ for every $x\in\XX_N$, $\xi\in\Xi$. 

To guarantee satisfaction of the terminal constraint, the 
horizon $N$ needs to be selected such that for every $x\in\CC$ and $\xi^N\in 
\Xi^{N+1}$, there 
exists 
$[u(0) \ldots  u(N-1)]$ such that for~\eqref{eq:LPV} with $x(0)=x$, $\xi(k)=\xi_k$ for all 
$k\in\zint{0,N}$, 
$(x(k),u(k))\in\CC_{xu}$  for all 
$k\in\zint{0,N-1}$,  and 
$x(N)\in\XX_N$.  
Let 
 \begin{align}\label{eq:horizonDesign}
 \SS^{(0)}&=\XX_N,\no \\
\SS^{(h+1)}_i&=\{x\in\XX:\ \exists u\in\UU, \no \ 
A_ix+Bu\in\SS^{(h)}\},\\
 \SS^{(h+1)}&=\bigcap_{i=1}^\ell \SS^{(h+1)}_i.
 \end{align} 
The set $\SS^{(h)}$ is such that for any $x(0)\in\SS^{(h)}$, given any 
$ \xi^{h-1}\in\Xi^h$, there exists a sequence $[u(0) \ldots  u(h-1)]$ such that 
for~\eqref{eq:LPV} 
with 
$x(0)=x$ and $\xi(k)=\xi_k$ for all 
$k\in\zint{0,N}$,  
$(x(k),u(k))\in\CC_{xu}$ and 
$x(h)\in\XX_N$. 

\begin{theorem}\label{th:constrStability1}
Consider~\eqref{eq:FHOCparam},  let
 $\bar h\in\zz_{0+}$ be such that $\CC^{(\bar 
h+1)}=\CC^{(\bar h)}=\CC$ in~\eqref{eq:mRCIcompute}, and let $\CC_{xu}$ be 
defined by~\eqref{eq:cxu}. 
Let $\XX_N=\XX^\infty$  from~\eqref{eq:MCAS}, where   $\XX_{xu}=\CC_{xu}$, 
and
 $N\in\zz_{0+}$ be such 
that $\SS^{(N)}\supseteq \CC$.  If $x(  t)\in\CC$ at $t\in\zz_{0+}$,  
 and    $\xi_\tau^N\in\Xi^{N+1}$, $\tilde 
\xi_{0|\tau}\in\tilde\Xi(\xi_{0|\tau})$ for all 
 $\tau \geq t$,  \eqref{eq:FHOCparam}~is feasible for all 
 $\tau \geq t$. 
If   there exists $t\in\zz_{0+}$  such that $  \xi_{k|\tau}= \bar \xi(\tau+k)$ for all 
$\tau\geq t$, $k\in\zz_{[0,N]}$,
\eqref{eq:system} in closed-loop with the 
MPC that solves~\eqref{eq:FHOCparam} is also AS in $\CC$.
\end{theorem}

\begin{proof}
Since $\CC\subseteq \SS^{(N)}$ for every $x(t)\in\CC$ and $ \xi_t^N\in\Xi^{N+1}$, there 
exists 
an 
input 
sequence  of length $N$ such that $(x_{k|t},u_{k|t})\in\CC_{xu}$ for all 
$k\in\zint{0,N-1}$, 
and   $x_{N|t}\in\XX_N$. Due to Lemma~\ref{lem:RCI},  for $u_{0|t}$ such that 
$(x_{0|t},u_{0|t})\in\CC_{xu}$, $\sum_{i=1}^\ell [  \xi_{0|t}+\tilde 
\xi_{0|t}]_iA_ix_{0|t}+Bu_{0|t}\in\CC$, 
for every  $\tilde \xi_{0|t}\in\tilde\Xi(\xi_{0|t}) $. 
Furthermore, if  $ \xi_{\tau}^N= \bar \xi_{\tau}^N\in\Xi^{N+1}$ for all $\tau\geq t$,  AS  in 
$\CC$ of 
\eqref{eq:system} in 
closed-loop 
with the MPC that solves~\eqref{eq:FHOCparam}
follows from Lemma~\ref{lem:ConstrLPVstable} noting that if $x_{N |t}\in\XX_N$, 
$x_{N-1|t+1}=x_{N |t}$, and  
$u_{N-1 |t+1}=\kappa(\xi_{N |t})x_{N |t}$, 
  then 
$(x_{N-1 |t+1},u_{N-1|t+1})\in\CC_{xu}$, i.e., \eqref{eq:parTermCtrl} is admissible in 
$\XX_N$ 
with 
respect~\eqref{eq:FHOCparam-RCI}, which follows from computing $\XX_N$ 
by~\eqref{eq:MCAS} with $\XX_{xu}=\CC_{xu}$.
\end{proof}

In the definition of $\SS^{(h)}$, i.e., \eqref{eq:horizonDesign}, the parameter sequence 
$\xi^h$ 
is 
known. This is due to enforcing  the terminal set only with respect to the 
nominal dynamics, while the robust invariance of~$\CC$ and choosing $N$ so that 
$\SS^{(N)}\supseteq 
\CC$ guarantee that at a successive 
step, even in presence of a parameter estimation error which causes a prediction error, the 
terminal set will still be reachable in $N$ steps.

\sdcRemoved{

}
There are alternative ways to compute $\CC$, other than as the maximal RCI. For instance, an RPI 
set 
can 
be constructed from~\eqref{eq:pLDIcl} by including additional constraints 
in~\eqref{eq:pLF-designLMI}. Then, for given $N\in\zz_+$,  $\CC$ can be obtained as $N$-step 
backward reachable set 
of such RPI. In this case the MPC horizon $N$ becomes a free design parameter. Such a
procedure is not fully described here due to the limited space. 

Theorem~\ref{th:constrStability1}  ensures robust feasibility 
of \eqref{eq:FHOCparam}, robust satisfaction of~\eqref{eq:sysConstr}, and 
nominal asymptotic stability, in the sense that the closed loop is AS if there exists $ 
t\in\zz_{0+}$ such that $\tilde \xi_{k|\tau}=0$, for all $\tau\geq t$, $k\in\zz_{[0,N]}$. 
Next, we 
combine 
Theorem~\ref{th:ISS} and~\ref{th:constrStability1}.

\section{Indirect-adaptive MPC: Complete Algorithm}\label{sec:IAMPC}

The last design element in~\eqref{eq:FHOCparam} is the construction of the parameter prediction 
vector  $\xi_t^N$. 

Since  $\bar \xi$ in~\eqref{eq:system} is assumed to be constant or slowly varying, an obvious 
choice would be $  \xi_{k|t}= \xi(t)$, for all $k\in\zint{0,N}$, for all $t\in\zz_{0+}$.  However, 
this choice violates the assumption of Theorem~\ref{th:ISS} (and implicitly 
those of Lemmas~\ref{lem:LPVstable} and~\ref{lem:ConstrLPVstable}) that requires 
$  \xi_{k |t+1}=   \xi_{k+1|t}$, for all $k\in\zint{0,N-1}$,  $t\in\zz_{0+}$. 
Such an assumption is required because  if  the  entire parameter  prediction 
vector $ \xi^N_t$ suddenly  changes,  the value function $\VV_N^\MPC$ may not be decreasing. 
This is due to using the pLF only as terminal cost, as opposed to enforcing it along the 
entire horizon~\cite{cuzzola02,mao03}, which then requires the solution of LMIs in real-time.

Thus, we introduce a $N$-step delay in the parameter prediction, 
\beq \label{eq:parUpdate}
   \xi_{k|t}=  \xi(t-N+k),\ \forall k\in\zint{ 0,N }.
  \eeq
Due to~\eqref{eq:parUpdate}, at each time $t$, the new estimate is 
added as last element of $\xi_t^N$, i.e., $  \xi_{N|t}= \xi(t)$ and $\xi_{k|t}=   
\xi_{k+1|t-1}$, 
for all 
$k\in\zint{0,N-1}$, $t\in\zz_{0+}$. We 
can now state the complete IAMPC strategy and its main result.

\begin{theorem}\label{th:final}
Consider~\eqref{eq:system}, where $\bar \xi \in \Xi$, in closed loop with the IAMPC 
controller 
that at every 
$t\in\zz_{0+}$ 
solves~\eqref{eq:FHOCparam}, where~$\PP(\xi )$  defined by~\eqref{eq:parLF} and 
$\kappa(\xi)$  defined 
by~\eqref{eq:parTermCtrl} are
from~\eqref{eq:pLF-designLMI}, $\CC$ and $\XX_N$ are designed according to  
Theorem~\ref{th:constrStability1} and $ 
\xi^N_t\in\Xi^{N+1}$ is 
obtained from~\eqref{eq:parUpdate}, 
where $  \xi(t)\in\Xi$ for all $t\in\zz_{0+}$. 
If for some $  t\in \zz_{0+}$, $x( t)\in\CC$,  the closed-loop   
satisfies~\eqref{eq:sysConstr}, 
and \eqref{eq:FHOCparam} is recursively 
feasible for any $\tau\geq t$. Furthermore, the closed loop is ISS in the RPI set 
$\CC$ with respect 
to $\tilde \xi_{0|t} = \bar \xi-\xi_{0|t}$, i.e., the $N$-steps delayed 
estimation error $\tilde \xi_{0|t}=\bar \xi - \xi(t-N)$.
\end{theorem}
\begin{proof}
The proof  follows by combining Theorem~\ref{th:ISS} with Theorem~\ref{th:constrStability1}. By 
Theorem~\ref{th:constrStability1}, $\CC$ is RCI, and if $x(t)\in\CC$, 
\eqref{eq:FHOCparam}~is  feasible  for all   $\tau\geq t$, for any
$\xi_\tau^N\in\Xi^{N+1}$ that satisfies~\eqref{eq:parUpdate}, since~\eqref{eq:parUpdate} 
implies that  $\xi_{k |\tau}=   \xi_{k +1|\tau-1}$, for all $k\in\zint{0,N-1}$. 
Thus, by~\eqref{eq:cxu} enforced in~\eqref{eq:FHOCparam}, $\CC\subseteq \XX$ is a compact 
RPI for the 
closed-loop system, and hence~\eqref{eq:sysConstr} is satisfied for all 
 $\tau\geq t$. 
Since $\VV_{\xi^N}^\MPC$ is piecewise quadratic for every $\xi^N\in\Xi^{N+1}$, by taking  
$\XX_\eta=\XX_L=\CC$, which is RPI for the 
closed-loop system and compact  since $\CC\subseteq \XX$, the existence of 
a Lipschitz  constant $L$ is guaranteed. 
Hence,
Theorem~\ref{th:ISS} holds within $\CC$,
proving ISS  with respect to $\tilde \xi_{0|t}
=\bar \xi-\xi_{0|t}= \bar \xi - 
\xi(t-N)$, i.e., the delayed estimation error. 
\end{proof}

Based on Theorem~\ref{th:final}, from any initial state $x( t)\in\CC$, the closed-loop 
system 
robustly satisfies the constraints for any admissible estimation error, and the expansion term 
in the ISS Lyapunov function is proportional to  the norm of the delayed parameter 
estimation error. Thus, if the parameter  estimate converges at time $t^*$ and such value 
is maintained for all $t\geq t^*$, for all $t\geq t^*+N$, $\tilde \xi_t^N=0$ and hence the 
closed-loop is AS. However, note that ISS holds regardless of such convergence.  Finally, 
note that 
at runtime, the IAMPC solves only a QP
as a standard (non-adaptive) linear MPC. Thus, the following corollary derives immediately 
from Theorem~\ref{th:final}.
\begin{corollary}
The IAMPC designed according to Theorem~\ref{th:final} solves 
Problem~\ref{prob:ampc}.
\end{corollary}
The requirement of soling only QP during execution in a significant reduction of computational 
load and code complexity with respect to robust MPC based on 
LMIs~\cite{KBM96,cuzzola02,mao03}. The drawback is on ensuring ISS versus the robust stability 
in~\cite{KBM96,cuzzola02,mao03}. However, here we still guarantee robust constraint 
satisfaction.  

\subsection{Comments on parameter estimator design}\label{ssec:Estimator}

The ISS  property established in  Theorem~\ref{th:final} implies that when the estimator 
converges to the true parameter value the closed-loop becomes AS, but it is more general than 
that. In fact, ISS ensures that, even if the estimate never converges, the expansion term in the 
Lyapunov function, and hence the ultimate bound on the state, is proportional the estimation error. 
Thus, ISS allows to state properties that ae parametric in the estimation error, and hence hold 
regardless of the convergence of the estimator. Because of this and because the IAMPC design 
does not require a specific choice for the estimator design, we have called the   
IAMPC as independent of the estimator. On the other hand, it is 
required for the estimator to provide 
$\xi(t)\in\Xi$, for all $t\in\zz_{0+}$, as per Assumption~\ref{ass:estimate}. To enforce 
Assumption~\ref{ass:estimate} one can 
always design an estimator that produces the unconstrained estimate $\varrho\in\rr^\ell$, while 
the 
IAMPC uses its 
projection onto $\Xi$, i.e., $ {\xi}={\rm proj}_\Xi(\varrho)$. By using the $\varrho\in\rr^\ell$  in 
the 
estimator update equation and providing ${\xi}={\rm proj}_\Xi(\eta)$ to the controller, this 
amounts to a standard estimator with an output nonlinearity. Thus, the convergence conditions will 
be the 
same as those for standard estimators, in particular  identifiability, and persistent excitation. 
Guaranteeing the persistent excitation in closed-loop systems is currently an active area of 
research also in MPC, see, e.g.,~\cite{marafioti2010persistently, rathousky2013mpc,WD14}. While 
it is certainly an interesting future research 
direction to merge the 
IAMPC developed here with some of the above techniques, it is worth remarking again that for 
the ISS property in Theorem~\ref{th:final} to hold, this is not necessary as convergence is not 
required.  

As regards to identifiability, a subject that is worth a brief discussion is the case 
where the true value of the parameter $\bar \xi$
is not uniquely define, which may be due to the polytopic representation~\eqref{eq:system} of the 
uncertain 
system. Given the actual system matrix $\bar A$ the set $\bar \xi(\bar A) =\{\bar 
\xi\in\Xi:\sum_{i=1}^\ell[\xi]_iA_i = \bar A\}$ may have cardinality greater than $1$. In 
this 
case 
we can provide a slightly modified ISS Lyapunov function.
\begin{corollary}
Let the assumptions of Theorem~\ref{th:ISS} hold and $\eps(\xi,{\bar A})  =\min_{\bar 
\xi\in\bar\xi(\bar A)}\|\xi-\bar\xi\|$. Then for~\eqref{eq:system} in closed loop with the MPC 
based 
on~\eqref{eq:FHOCparam}, $\VV_{\xi_{t+1}^N}^\MPC(x(t+1))\leq  \VV_{\xi_{t }^N}^\MPC(x(t)) 
-\lambda_\rmin(Q) 
\|x(t)\|^2 +   \gamma_{\rm ISS}\cdot \eps(\xi,{\bar A})$, i.e., $\VV_{\xi^N}^\MPC$ is an ISS 
Lyapunov function  
  with respect to $\eps(\xi,{\bar A})$. 
\end{corollary}
The proof follows directly by the fact that Theorem~\ref{th:ISS} and all subsequent results  only 
use the 
difference between the predicted and actual system state, which is the same for all 
$\bar\xi\in\xi(A)$. Thus~\eqref{eq:issLF} 
holds for 
all  values $\bar \xi\in\xi(A)$, which means that it has to hold for the smallest expansion term, 
which is $\eps(\xi,{\bar A})$. Loosely, this means that while formulated on the convex 
combination vector for computational purposes, the  expansion 
term is a function of the difference between the estimated and 
actual system matrix. 
%
%

\bfig[!t]
\bcnt 
\psfrag{k}{$~~~t$}
\psfrag{x1,x2}{$x_1,x_2$}
\psfrag{x1}{$x_1$}
\psfrag{x2}{$x_2$}
\psfrag{u}{$u$}
{\image{.68}{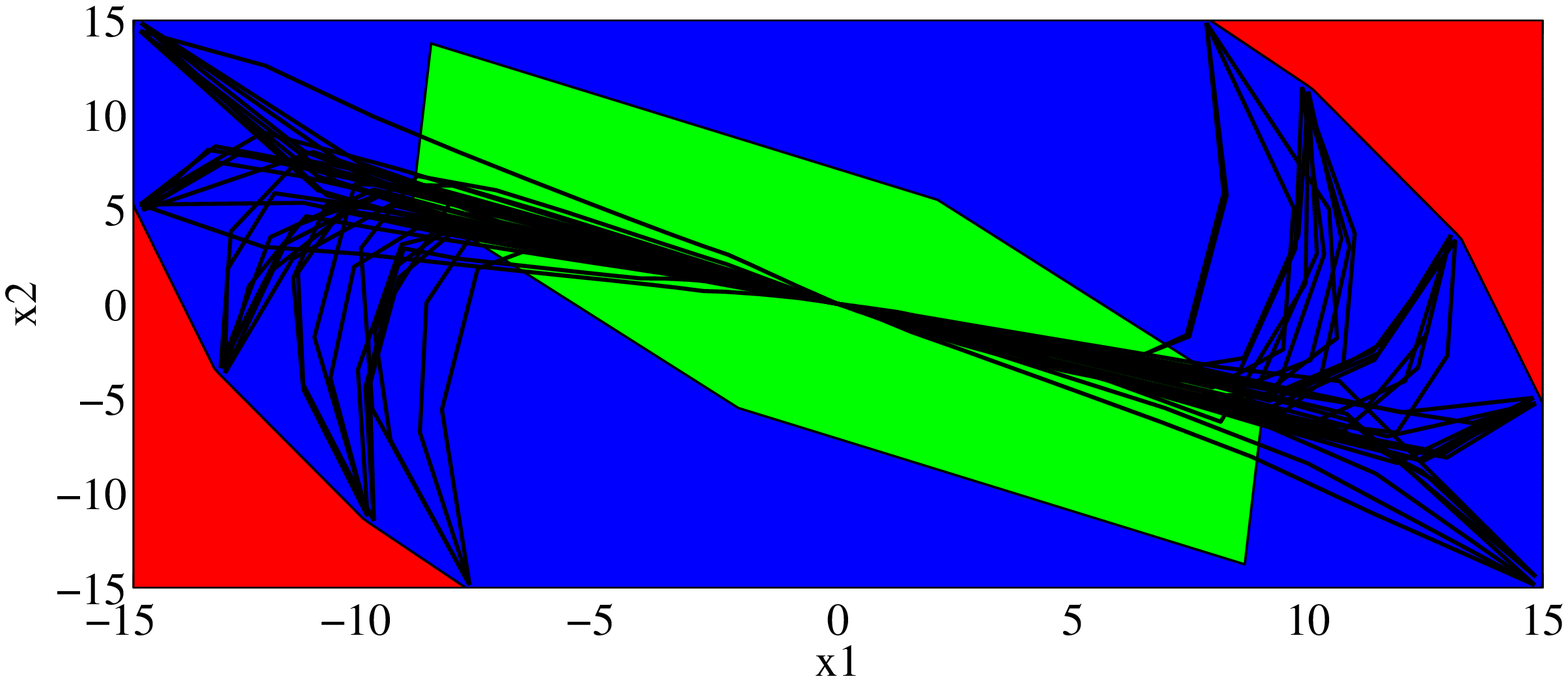}}
\caption{IAMPC simulations on the numerical example. Trajectories (black), $\XX^\infty$ (green), 
$\CC$ (blue), $\XX$ 
(red).}\label{fig:sim-mRCI}
\ecnt
\efig 

\section{Numerical Simulations}\label{sec:sims}
First, we show some simulations on a numerical example.
We consider~\eqref{eq:system}, where $\ell = 5$, and  $A_1=\bsm 1 
& 0.2 \\ 0 & 1\esm 
$, $A_2=1.1\cdot A_1$,
$A_3=0.6\cdot A_1$,  $A_{4} = \bsm 0.9 & 0.3 \\  0.4 & 0.6\esm$ 
 $A_{5} = \bsm  0.95 & 0 \\ 0.8 & 1.02\esm$, and $B=\bsm -0.035 &  -0.905\esm'$. While being 
only second order system, this example is challenging because some of the dynamics are stable 
and some unstable, and the system matrices are in some cases significantly different. In fact, for a 
similar 
system, \cite{dic15}~showed that without proper cost function adaptation, the closed-loop 
may not be AS even when the perfect model is estimated. The constraints are 
defined by~\eqref{eq:sysConstr}, where $\XX=\{x\in\rr^2:\ |[x]_i|\leq 15,\ i=1,2\}$, 
$\UU=\{u\in\rr:\ |[u]|\leq 10 \}$. We have implemented a simple estimator that computes the 
least 
squares  solution $\varrho(t)$ based on past data window of $N_m$ steps and applies a first order 
filter on the projection of $\varrho(t)$ onto $\Xi$, i.e., $\xi(t+1)= 
(1-\varsigma)\xi(t)+\varsigma\cdot {\rm proj}_\Xi(\varrho(t))$, where
$\varsigma\in\rr_{(0,1)}$, and
$[\xi(0)]_i=1/\ell$, $i\in\zint{1,\ell}$. Such simple estimator satisfies 
Assumption~\ref{ass:estimate} because projection, summation and the guarantee that the result 
is a convex combination vector. Also, the least square problem can be regularized by a term 
based on 
$\varrho(t-1)$, and the projection can be computed by solving a simple QP. In the simulations we 
use the QP solver in~\cite{pqp}, for both projection and MPC control computation. We design the 
controller according to Theorem~\ref{th:constrStability1}, where 
$\CC=\CC^\infty$, and we select $N_m=3$ and $N=8$, which is the smallest value 
such that  $\SS^{(N)}\supseteq 
\CC^\infty$   by~\eqref{eq:horizonDesign}. Figure~\ref{fig:sim-mRCI}  shows the 
simulations where the initial conditions
are the 
vertices of 
$\CC$ and for each initial condition, $4$ different simulations with different (random) values of 
$\bar \xi\in\Xi$ are executed. 
Figure~\ref{fig:sim-ISS}  compares the cases 
where $\varsigma = 1/2$ and $\varsigma = 1/20$, i.e., fast versus slow estimation, thus 
showing the impact of the estimation error on the closed-loop behavior.


\bfig
\bcnt 
\psfrag{k}{$~~~t$}
\psfrag{x1,x2}{$x_1,x_2$}
\psfrag{x1}{$x_1$}
\psfrag{x2}{$x_2$}
\psfrag{u}{$u$}
{\image{.78}{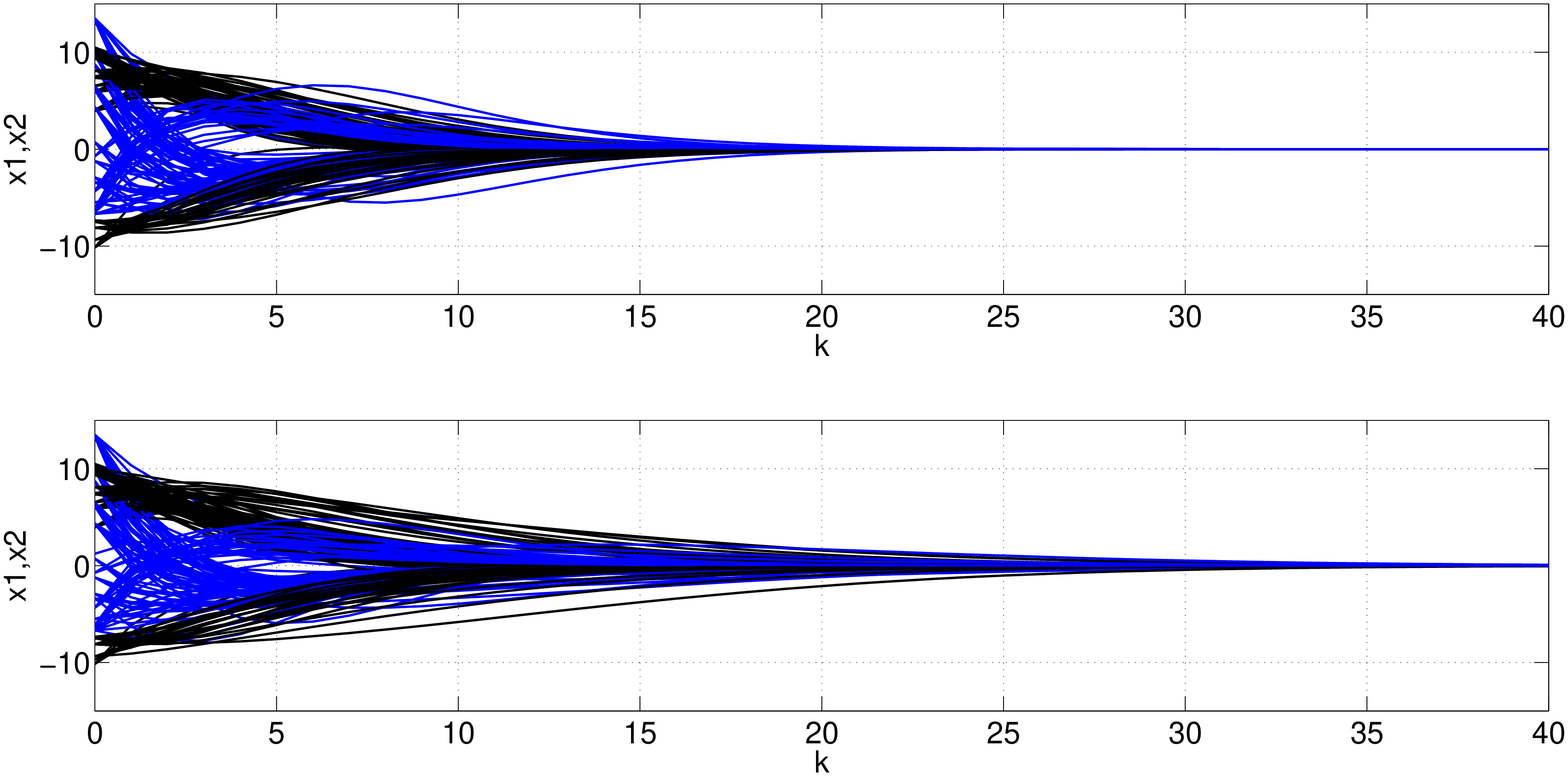}}
\caption{Simulations with fast ($\varsigma = 1/2$, upper plot) and slow ($\varsigma = 
1/20$, upper 
plot) 
parameter estimator, $\lbrack x\rbrack_1$ blue, $\lbrack x\rbrack_2$ 
black.}\label{fig:sim-ISS}
\ecnt
\efig 

\section{Conclusions and Future Work}\label{sec:conclus}
We have proposed an indirect-adaptive MPC that guarantees robust constraint satisfaction, 
recursive 
feasibility, and ISS with respect to the parameter estimation error, yet has computational 
requirements similar to standard MPC.

The IAMPC can be easily modified to handle uncertainty  also in the 
input-to-state matrix $B$, to exploit non-maximal  yet faster to compute  RCI sets, and to account 
for additive disturbances. Future works will detail these, as well considering tracking  and 
designs resulting in a different ISS expansion term providing AS even in the presence of a 
small-but-non-zero
error in the parameter estimate.


\bibliographystyle{IEEEtran}
\bibliography{IAMPC_pLF-bib}

\end{document}